\documentclass[12pt,a4paper]{article}
\usepackage{amsmath,amsthm,amsfonts,amssymb,bbm}
\usepackage{graphicx,psfrag,subfigure,color}
\usepackage{cite}
\usepackage{hyperref}

\numberwithin{equation}{section}

\newcommand{\Pb}{\mathbbm{P}}
\newcommand{\E}{\mathbbm{E}}
\newcommand{\Id}{\mathbbm{1}}

\newcommand{\I}{{\rm i}}
\newcommand{\D}{\mathrm{d}}

\newcommand{\R}{\mathbb{R}}
\newcommand{\N}{\mathbb{N}}
\newcommand{\HH}{\mathbb{H}}
\newcommand{\CH}{\mathcal{H}}
\newcommand{\K}{\mathbb{K}}
\newcommand{\Z}{\mathbb{Z}}

\renewcommand{\Im}{\mathrm{Im}}

\newcommand{\WW}{\circ}
\newcommand{\BB}{\bullet}

\newtheorem{prop}{Proposition}[section]
\newtheorem{thm}[prop]{Theorem}
\newtheorem{lem}[prop]{Lemma}

\newtheorem{cor}[prop]{Corollary}

\newtheorem{cla}[prop]{Claim}

\newtheorem{rem}[prop]{Remark}

\title{A combinatorial identity for the speed of growth in an anisotropic KPZ model}
\date{27. January 2016}
\author{Sunil Chhita\thanks{Institute for Applied Mathematics, Bonn University, Endenicher Allee 60, 53115 Bonn, Germany. E-mail: {\tt schhita@iam.uni-bonn.de}}\and Patrik L.\ Ferrari\thanks{Institute for Applied Mathematics, Bonn University, Endenicher Allee 60, 53115 Bonn, Germany. E-mail: {\tt ferrari@uni-bonn.de}}}

\begin{document}

\maketitle
\sloppy
\begin{abstract}
The speed of growth for a particular stochastic growth model introduced by Borodin and Ferrari in~\cite{BF08}, which belongs to the KPZ anisotropic universality class, was computed using multi-time correlations. The model was recently generalized by Toninelli in~\cite{Ton15} and for this generalization the stationary measure is known but the time correlations are unknown. In this note, we obtain algebraic and combinatorial proofs for the expression of the speed of growth from the prescribed dynamics.
\end{abstract}

\newpage
\section{Introduction}\label{SectIntro}
This note considers a stochastic growth model in the \emph{KPZ anisotropic class} in \mbox{$2+1$ dimensions}. This model was introduced in~\cite{BF08} and studied  in depth for a specific initial condition, the case considered here.  This model describes the evolution of particles subject to an \emph{interlacing property}.  The model can also be thought of as a two-dimensional \emph{stochastically growing interface}, or as a \emph{random lozenge tiling model}. Another tiling model which shares similar features of the dynamical perspective is \emph{domino tilings of the Aztec diamond} using \emph{the shuffling algorithm}; see~\cite{Nor08,BF15}. For the model considered in this paper, the evolution of the interface, under hydrodynamic scaling, grows deterministically according to a PDE. At the microscopic level in the bulk, the specified boundary conditions of the system are \emph{forgotten}: in the bulk of the system one sees an invariant measure which depends only on the normal direction of the macroscopic surface~\cite{BF08}. From the lozenge tiling perspective, these limiting measures are determinantal and they are parameterized by the relative proportions of lozenges. These measures match up with the dimer model on the infinite honeycomb graph which are the unique translation invariant stationary measures for any given normal direction~\cite{KOS03,Shef05}.

Stochastic growth models have been studied in many different guises. Many of these studies have focussed around the 1+1 KPZ universality class; see for example~\cite{Cor11,Q12,QS15,Fer10b,SS10c} for surveys. The $d=2$ anisotropic case has not been as extensively studied~\cite{Wol91,PS97,BF08,Ku11} but there are, however, numerous results in connection with random tiling models~\cite{BF15,Nor08,BG08,Du13,Pet12,Pet12b,Ken04}. One feature of these models which is of particular interest for this note is the \emph{limit shape}~\cite{CKP01,KO07}. This is the average profile which the stochastic interface fluctuates around.  In particular, we focus on giving an elementary approach to computing the \emph{speed of growth}, which is the growth rate of the stochastic growing interface under the prescribed dynamics. This is an important quantity since it determines the limit shape of the system.

In the work by Borodin and Ferrari~\cite{BF08}, the speed of growth was obtained by computing the \emph{infinitesimal current} and then taking the bulk scaling limit. The computation was relatively straightforward, but requires the knowledge of the correlations of particles or lozenges at two different times.   Recently Toninelli in~\cite{Ton15}, considered the same model and a generalized version of the dynamics for the infinite honeycomb graph using a \emph{bead} perspective; beads and the \emph{bead model} were introduced in~\cite{Bou09}. He shows that the model is well-defined for the stationary measure. This requires some effort since the dynamics allow, a-priori, infinitely long-range interactions, but under the stationary measure these interaction probabilities decay exponentially with the distance. The speed of growth was not determined in~\cite{Ton15}.

By appealing to the underlying combinatorics of the model and the Kasteleyn approach for dimer models (e.g. see~\cite{KenLectures}), we are able to determine the speed of growth in the infinite honeycomb case, thus establishing the conjecture in~\cite[Eq. (3.6)]{Ton15}\footnote{ArXiv version 1}; see Theorem~\ref{ThmInfiniteSystem}. From~\cite{BF08}, the speed of growth $v$ is given by a specific (unsigned) off-diagonal entry of the inverse of the Kasteleyn matrix, where the Kasteleyn matrix is a type of (possibly signed) adjacency matrix~\cite{Kas61}. The approach used is to first find a recursive formula for this particular entry for a particular subgraph of the honeycomb graph and then to extend this subgraph to the infinite plane using known results in the literature. The resulting limiting series matches with the speed of growth defined through the dynamics under the stationary measure. To prove Theorem~\ref{ThmInfiniteSystem}, we found a relation for the specific off-diagonal entry of the inverse of the Kasteleyn matrix for any \emph{tileable} finite honeycomb graph with arbitrary edge weights in terms of single times. As a consequence, we use this relation to motivate an algebraic proof that the expression for the speed of growth of the model in~\cite{BF08} and the speed of growth computed from the prescribed dynamics are the same (the latter is a series with entries given by determintants of increasing size). In principle it seems to be feasible to obtain Theorem~\ref{ThmInfiniteSystem} from Theorem~\ref{ThmFiniteSystem} by a precise asymptotic analysis and a careful manipulation of sums, but we did not pursue this since we are primarily interested in understanding combinatorial structures behind the identity.

\bigskip\noindent
{\bf Acknowledgments.} The authors are grateful for discussions with F. Toninelli about his work and to both ICERM and the Galileo Galileo Institute which provided the platform to make such discussions possible. The work is supported by the German Research Foundation via the SFB 1060--B04 project.

\section{Model and results}\label{SectModel}

\subsection{The finite particle model}\label{sectFiniteT}
We first describe the interacting particle system model introduced in~\cite{BF08}. It is a model in the $2+1$-dimensional KPZ anisotropic class. It is a continuous time Markov chain on the state space of interlacing variables
\begin{equation}
{\rm GT}_N=\Big\{\{x_k^m\}_{\begin{subarray}{ll} k=1,\dots,m\\m=1,\dots,N\end{subarray}}
\subset \Z^{\frac{N(N+1)}2}\mid x^m_{k-1}<x_{k-1}^{m-1}\leq x_k^m\Big\}.
\end{equation}
$x_k^m$ is interpreted as the position of the particle with label $(k,m)$, but we will also refer to a given particle as $x_k^m$. We consider fully-packed initial conditions, namely at time moment $t=0$ we have $x_k^m(0)=k-m-1$ for all $k,m$; see Figure~\ref{FigureInitialConditions}.
\begin{figure}
\begin{center}
\psfrag{x}[b]{$x$}
\psfrag{n}[b]{$n$}
\psfrag{h}[b]{$h$}
\psfrag{x11}[cl]{$x_1^1$}
\psfrag{x12}[cl]{$x_1^2$}
\psfrag{x13}[cl]{$x_1^3$}
\psfrag{x14}[cl]{$x_1^4$}
\psfrag{x15}[cl]{$x_1^5$}
\psfrag{x22}[cl]{$x_2^2$}
\psfrag{x23}[cl]{$x_2^3$}
\psfrag{x24}[cl]{$x_2^4$}
\psfrag{x25}[cl]{$x_2^5$}
\psfrag{x33}[cl]{$x_3^3$}
\psfrag{x34}[cl]{$x_3^4$}
\psfrag{x35}[cl]{$x_3^5$}
\psfrag{x44}[cl]{$x_4^4$}
\psfrag{x45}[cl]{$x_4^5$}
\psfrag{x55}[cl]{$x_5^5$}
\includegraphics[height=4.3cm]{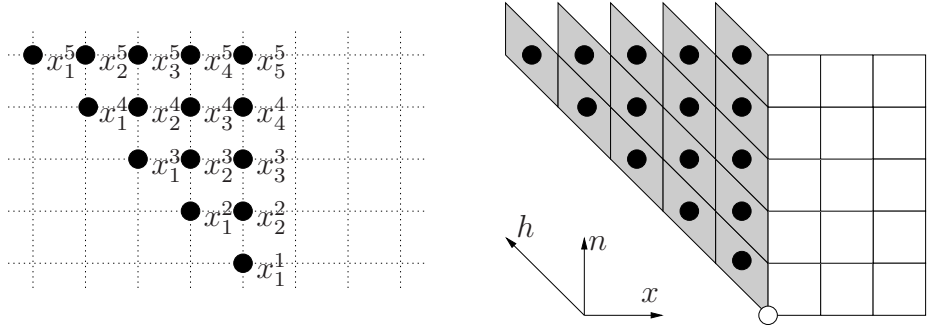}
\caption{Illustration of the initial conditions for the particles system and the corresponding lozenge tilings. In the height function picture, the white circle has coordinates \mbox{$(x,n,h)=(-1/2,-1/2,0)$}.}
\label{FigureInitialConditions}
\end{center}
\end{figure}

The particles evolve according to the following dynamics. Each particle $x_k^m$ has an independent exponential clock of rate one, and when the $x_k^m$-clock rings the particle attempts to jump to the right by one. If at that moment \mbox{$x_k^m=x_k^{m-1}-1$} then the jump is blocked. Otherwise, we find the largest $c\geq 1$ such that \mbox{$x_k^m=x_{k+1}^{m+1}=\dots=x_{k+c-1}^{m+c-1}$}, and all $c$ particles in this string jump to the right by one.

We illustrate the dynamics using Figure~\ref{FigureIntro},
which shows a possible configuration of particles obtained from the fully-packed
initial condition. In this state of the system, if the $x_1^3$-clock
rings, then the particle $x_1^3$ does not move, because it is blocked by
particle $x_1^2$. If the $x_2^2$-clock rings, then
the particle $x_2^2$ moves to the right by one unit, but to respect the interlacing property, the particles $x_3^3$ and $x_4^4$ also
move by one unit to the right at the same time. This aspect of the dynamics is
called \emph{pushing}.
\begin{figure}
\begin{center}
\psfrag{x}[b]{$x$}
\psfrag{n}[b]{$n$}
\psfrag{h}[b]{$h$}
\psfrag{x11}[cl]{$x_1^1$}
\psfrag{x12}[cl]{$x_1^2$}
\psfrag{x13}[cl]{$x_1^3$}
\psfrag{x14}[cl]{$x_1^4$}
\psfrag{x15}[cl]{$x_1^5$}
\psfrag{x22}[cl]{$x_2^2$}
\psfrag{x23}[cl]{$x_2^3$}
\psfrag{x24}[cl]{$x_2^4$}
\psfrag{x25}[cl]{$x_2^5$}
\psfrag{x33}[cl]{$x_3^3$}
\psfrag{x34}[cl]{$x_3^4$}
\psfrag{x35}[cl]{$x_3^5$}
\psfrag{x44}[cl]{$x_4^4$}
\psfrag{x45}[cl]{$x_4^5$}
\psfrag{x55}[cl]{$x_5^5$}
\includegraphics[height=4.3cm]{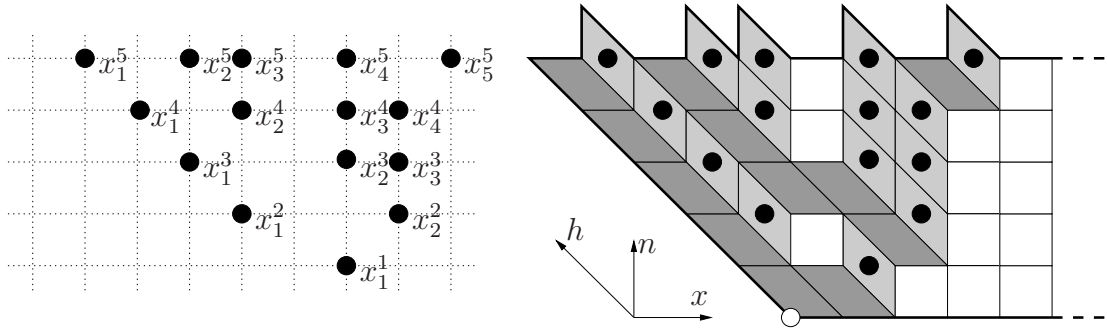}
\caption{From particle configurations (left) to 3d visualization via lozenge tilings (right).}
\label{FigureIntro}
\end{center}
\end{figure}

\begin{rem}\label{remLozenge}
The positions of the particles uniquely determine a lozenge tiling in the region bordered by the thick line in Figure~\ref{FigureIntro}.
\end{rem}

\begin{rem}\label{remMeasureUnif}
As shown in~\cite{BF08}, the measure at time $t$ generated by the dynamics starting from the fully-packed initial condition has the property that, conditioned on the measure of the particles at level $N$, the other particles are uniformly distributed on ${\rm GT}_N$ with fixed level $N$ configuration.
\end{rem}

Further, it is proved in Theorem~1.1 of~\cite{BF08} that the correlation function of the particles are determinantal on a subset of space-time~\cite{BF07} (see~\cite{Lyo03,BKPV05,Sos06,Jo05,Spo05,Bor09,RB04} for information on determinantal point processes). Denote by $\eta(x,n,t)$ the random variable that is equal to $1$ if there is a particle at $(x,n)$ at time $t$ and $0$ otherwise.

\begin{thm}[Theorem~1.1 of~\cite{BF08}]\label{ThmDetStructure}
For any given \mbox{$m\in \N$}, consider \mbox{$t_1\leq t_2\leq \dots\leq t_m$}, and $n_1\geq n_2\geq \dots\geq n_m$. Then
\begin{equation}
\Pb\left[\cap_{i=1}^m\{\eta(x_i,n_i,t_i)=1\}\right]=\det{[{\cal K}(x_i,n_i,t_i;x_j,n_j,t_j)]}_{i,j=1}^m,
\end{equation}
where
\begin{multline}\label{StartingKernel}
{\cal K}(x_1,n_1,t_1;x_2,n_2,t_2) =-\frac{1}{2\pi\I}\oint_{\Gamma_0}
\frac{\D w}{w^{x_2-x_1+1}}\,\frac{e^{(t_1-t_2)/w}}
{(1-w)^{n_2-n_1}}\,\Id_{[(n_1,t_1)\prec (n_2,t_2)]} \\
+\frac{1}{(2\pi i)^2}\oint_{\Gamma_0}\D w\oint_{\Gamma_1}\D z
\frac{e^{t_1/w}}{e^{t_2/z}} \frac{(1-w)^{n_1}}{(1-z)^{n_2}}
\frac{w^{x_1}}{z^{x_2+1}}\,\frac{1}{w-z},
\end{multline}
the contours $\Gamma_0$, $\Gamma_{1}$ are simple positively oriented
closed paths that include the poles $0$ and $1$, respectively, and
no other poles (hence, they are disjoint).
Here we used the notation
\begin{equation}\label{eqPartialOrder}
(n_1,t_1)\prec (n_2,t_2) \quad\text{iff} \quad n_1\leq n_2,
t_1\geq t_2,\text{ and }(n_1,t_1)\neq (n_2,t_2).
\end{equation}
\end{thm}

Given a configuration of particles, we define the height $h(x,n,t)$ as
\begin{equation}\label{eqDefinHeight}
h(x,n,t)=\#\{k\in\{1,\dots,n\}\mid x_k^n(t)>x\}.
\end{equation}
In particular, the growth rate of the height at a position $(x,n)$ is given by the (infinitesimal) particle current at $(x,n)$, denoted by $j(x,n,t)$ with
\begin{equation}\label{eq2.6}
j(x,n,t)=\lim_{\epsilon\to 0}\epsilon^{-1}\E\left[\eta(x,n,t)(1-\eta(x,n,t+\epsilon)\right].
\end{equation}

This quantity was computed in the proof of Lemma~5.3 of~\cite{BF08} with the result
\begin{equation}
j(x,n,t)={\cal K}_t(x,n;x+1,n),
\end{equation}
where here we use the notation ${\cal K}_t(x,n;y,m):={\cal K}(x,n,t;y,m,t)$.

On the other hand, at any time $t$, the height function $h$ at $(x,n)$ increases by one whenever there is a particle at position $(x,n)$ which jumps to \mbox{$(x+1,n)$}. By definition of the dynamics,  at rate 1 a particle at $(x,n)$ could jump to the right provided that there is no particle at $(x+1,n-1)$. However, a particle at $(x,n)$ can also be pushed if there is a column of particles directly below $(x,n)$. More explicitly, for $\ell>0$ a particle at $(x,n)$ is pushed to the right by the move of a particle at $(x,n-\ell+1)$. This happens at rate $1$ provided that $(x+1,n-\ell)$ is empty and $(x,n-k)$ for $k=0,\ldots,\ell-1$ are occupied. In the case $\ell=n$, the constraint that $(x+1,0)$ is empty does not exist as there are no particles at level $0$. Therefore the growth rate obtained from the dynamics, denoted by $v(x,n,t)$, is given by
\begin{equation}\label{eqspeed}
v(x,n,t):=\sum_{\ell=1}^{n-1} \E\left[(1-\eta(x+1,n-\ell,t))\prod_{k=0}^{\ell-1} \eta(x,n-k,t)\right]+\E\left[\prod_{k=0}^{n-1} \eta(x,n-k,t)\right].
\end{equation}

Using Theorem~\ref{ThmDetStructure} and the complementation principle for determinantal point processes (see Appendix of~\cite{BOO00}), the expected value in the right side of (\ref{eqspeed}) is given by
\begin{equation}
\det\left[
\begin{array}{cc}
\left[{\cal K}_t(x,n-i;x,n-j)\right]_{i,j=0}^{\ell-1} & \left[-{\cal K}_t(x,n-i;x+1,n-\ell)\right]_{i=0}^{\ell-1}\\
\left[{\cal K}_t(x+1,n-\ell;x,n-j)\right]_{j=0}^{\ell-1} & 1-{\cal K}_t(x+1,n-\ell;x+1,n-\ell)
\end{array}
\right].
\end{equation}

In this paper we show directly that indeed $v$ and $j$ are the same.
\begin{thm}\label{ThmFiniteSystem}
It holds
\begin{equation}
j(x,n,t)=v(x,n,t).
\end{equation}
\end{thm}
\begin{rem}
This finite-time result can be easily generalized to the case of level-dependent jump rates. This system is still determinantal when starting with fully-packed initial conditions, with correlation kernel obtained in Proposition 3.1 of~\cite{BF07}; see Corollary~2.26 of~\cite{BF08} too.
\end{rem}

\subsection{Particle, lozenge and dimer representations}\label{SectRepresentations}
The height function, defined in~\eqref{eqDefinHeight}, gives a three dimensional surface  with Cartesian coordinate axis.
In a projection in the $(1,1,1)$-direction, each unit square of the surface becomes a lozenge, while in the projection of Figure~\ref{FigureIntro} there are three types of parallelograms, still referred to as lozenges; see Figure~\ref{FigureTilings}.
\begin{figure}
\begin{center}
\psfrag{Tiles}{Facets}
\psfrag{Lozenges}{Lozenges}
\psfrag{Dimers}{Dimers}
\psfrag{type}{Type - Weight}
\psfrag{I}[c]{I - $b$}
\psfrag{II}[c]{II - $a$}
\psfrag{III}[c]{III - $c$}
\includegraphics[height=6cm]{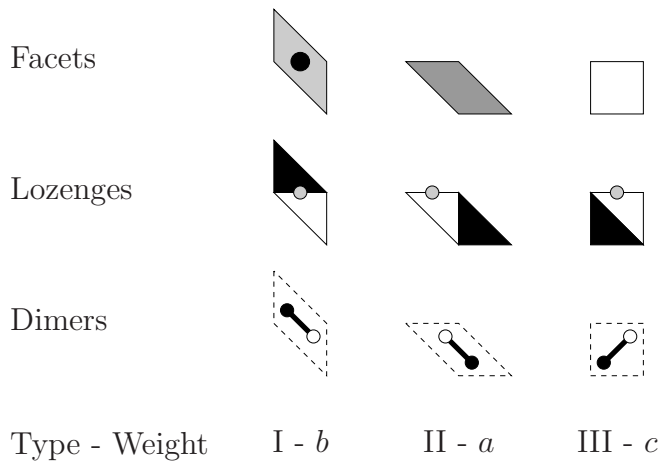}
\caption{Figure~\ref{FigureIntro} facets' types and their associated lozenges and angles. The gray circle is at position $(x,n)$. For dimers, their coordinate is given by the black site. More precisely, in this case we say that dimer of type I is at $(x,n)$, of type II is at $(x+1,n-1)$, while of type III at $(x,n-1)$.}
\label{FigureTilings}
\end{center}
\end{figure}
Another representation useful for the proofs of the theorems in this paper, is through \emph{perfect matchings} or \emph{dimers} via the dual graph associated to the underlying graph of the lozenge tilings.    For a lozenge tiling on a finite graph, the dual graph is a subgraph of the \emph{honeycomb} (or \emph{hexagonal}) graph, with each lozenge representing an edge, called a dimer, on this dual graph; see Figure~\ref{FigureTilings}.  Each lozenge tiling represents a \emph{dimer covering} of the dual graph; see Figure~\ref{FigureDimers} for an example.
\begin{figure}
\begin{center}
\psfrag{a}[r]{$(x,n)$}
\psfrag{b}[l]{$(x+1,n)$}
\psfrag{c}[r]{$(x+1,n-1)$}
\psfrag{d}[l]{$(x,n+1)$}
\includegraphics[height=3.5cm]{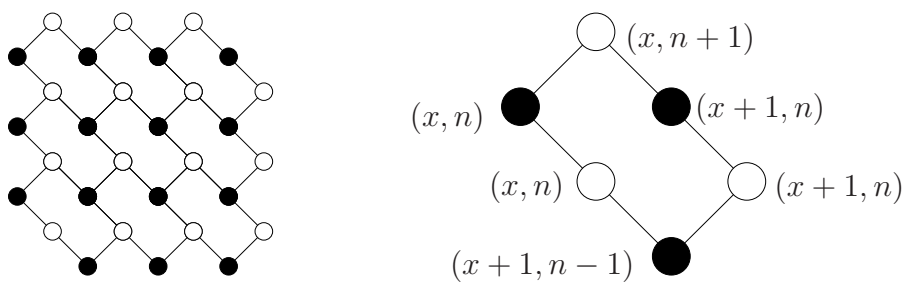}
\caption{Bipartite honeycomb graph for the dimer model and coordinate system of black and white vertices.}
\label{FigCoordinates}
\end{center}
\end{figure}
For the purpose of this paper, we denote $\mathcal{H}$ to be the infinite (bipartite) honeycomb graph; see the left side of Figure~\ref{FigCoordinates} for a finite snapshot.  We set $\mathcal{H}_L=\mathcal{H} \slash L$, that is, the restriction of length-size $L$ of $\mathcal{H}$ with periodic boundary conditions.  We use the standard terminology for the dimer model; see for example Section 1 in~\cite{KenLectures} for details.
\begin{figure}
\begin{center}
\includegraphics[height=3.5cm]{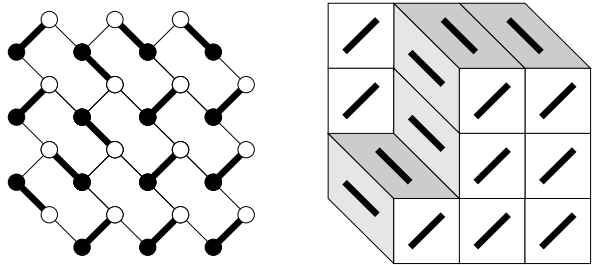}
\caption{From a dimer (thick lines) on the bipartite honeycomb graph to a lozenge configuration.}
\label{FigureDimers}
\end{center}
\end{figure}

\subsection{The model on $\Z^2$}
Now we consider the case where the state space are particles on $\Z^2$ satisfying interlacing between levels $n$ and $n+1$, for any $n\in\Z$, i.e.,
\begin{equation}
{\cal G}=\{x_k^n\in\Z, k,n\in\Z \,|\, x^{n+1}_k< x_k^n\leq x_{k+1}^{n+1},k,n\in\Z\}.
\end{equation}
Formally one would like to consider the dynamics on ${\cal G}$ as described for the model in Section~\ref{sectFiniteT}.  However, the dynamics are not well-defined for all elements in $\cal G$ due to the possibility of pushing from an infinitely long stack of particles.  It is shown in~\cite{Ton15}, that the dynamics are almost surely well-defined starting from the stationary measure because under this measure, the probability of having a stack of particles of length $\ell$ decays exponentially in $\ell$.

The translation invariant stationary measures form a two-parameter family, uniquely determined by the ``two-dimensional slope'' of the height function. In terms of dimers, we give to each type of dimer a weight, say $a,b,c$ as indicated in Figure~\ref{FigureTilings}. The probability of a dimer configuration is then proportional to the product of the weights of each dimer configuration, thus we effectively have  only two free parameters.

One nice property of stationary measures on $\Z^2$ is that their correlation functions are determinantal with a correlation kernel given only in  terms of the slope (which is determined by $a,b,c$). Correlations functions, at different times for the model on $\Z^2$ with stationary initial conditions, are not explicitly known (and for the asymmetric version studied by Toninelli~\cite{Ton15} they are unlikely to be determinantal) and therefore one cannot use (\ref{eq2.6}) to determine the speed of growth. However, the speed of growth is defined by the dynamics via (\ref{eqspeed}) with $n$ replaced by $\infty$.

Sheffield showed that indeed the translation invariant measures are uniquely determined by the slope of the height function.
\begin{thm}[Sheffield~\cite{Shef05}]\label{ThmSheffield}
For each $\nu = (p_a, p_b, p_c)$ with $p_a, p_b, p_c \geq 0$ with $p_a+p_b+p_c=1$ there is a unique translation-invariant ergodic Gibbs measure $M_\nu$ on
the set of dimer coverings of $\cal H$, for which the height function has average normal $\nu$. This measure can be obtained as the limit as $L\to\infty$ of the uniform measure on the set of those dimer coverings of $\CH_L$, whose proportion of dimers in the three orientations is $(p_a : p_b : p_c)$, up to errors tending to zero as $L\to\infty$. Moreover every ergodic Gibbs measure on $\cal H$ is of the above type for some $\nu$.
\end{thm}

The correlation functions of the stationary measure $M_\nu$ are determinantal and they are given in terms of the so-called inverse Kasteleyn matrix, denoted by $\K_\nu^{-1}$ (see more details in Section~\ref{Sect3.2} and the lecture notes~\cite{KenLectures} for a complete treatment of the subject). It is given as follows (it is obtained from Eq.~(4) in~\cite{Ken04} with appropriate change of coordinates)
\begin{equation}\label{eq2.12}
\K_\nu^{-1}(x,n;x',n')=b \left(\frac{a}{c}\right)^{x-x'}\left(\frac{b}{c}\right)^{n-n'} \K^{-1}_{abc}(x,n;x',n')
\end{equation}
with
\begin{equation} \label{eq2.13}
\K^{-1}_{abc}(x,n;x',n') = \frac{1}{(2\pi\I)^2}\oint_{|z|=1}dz\oint_{|w|=1}dw \frac{z^{n-n'} w^{n'-n+x'-x-1}}{a+bz+cw}.
\end{equation}
As shown by Kenyon, Okounkov and Sheffield in~\cite{KOS03}, the latter is the limiting inverse Kasteleyn matrix obtained in the toroidal exhaustion limit where the edge weights $a,b$ and $c$ are depicted in Figure~\ref{FigureTilings}.

Kenyon introduced a very useful mapping from $(p_a,p_b,p_c)$ to the upper-half complex plane $\HH$ illustrated in Figure~\ref{FigGeometry}.
\begin{figure}[h!]
\begin{center}
\psfrag{ba}[l]{$b/a$}
\psfrag{ca}[r]{$c/a$}
\psfrag{pa}{$\theta_a$}
\psfrag{pb}{$\theta_b$}
\psfrag{pc}{$\theta_c$}
\psfrag{Omega}[c]{$\Omega_{abc}$}
\psfrag{0}[c]{$0$}
\psfrag{1}[c]{$1$}
\includegraphics[height=4cm]{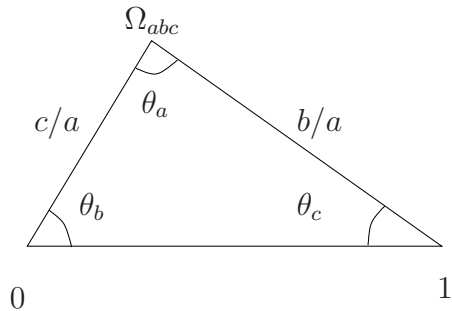}
\caption{Point $\Omega_{abc}\in\HH$ associated with $(p_a,p_b,p_c)$ is given by constructing the triangle with basis $\overline{01}$ and angles $\theta_k=\pi p_k$, $k\in\{a,b,c\}$ as indicated. The length of the segment $\overline{0\Omega_{abc}}$ is $c/a$ and the length of the segment of $\overline{1\Omega_{abc}}$ is $b/a$.}
\label{FigGeometry}
\end{center}
\end{figure}
By change of variables $z\to z a/b$ and $w\to w a/c$ in Eq.~\eqref{eq2.13} one obtains
\begin{equation}\label{eq2.14}
\K^{-1}_\nu(x,n;x',n')=\frac{1}{(2\pi\I)^2}\oint_{|z|=b/a}dz\oint_{|w|=c/a}dw \frac{z^{n-n'} w^{n'-n+x'-x-1}}{1+z+w}.
\end{equation}
A simple computation (see Appendix) leads to the following representation of the kernel
\begin{equation}\label{eqApp1}
\K^{-1}_\nu(x,n;x',n')=\frac{(-1)^{n-n'+x-x'}}{2\pi\I}\int_{\overline\Omega_{abc}}^{\Omega_{abc}} dw \frac{(w-1)^{n-n'}}{w^{n-n'+x-x'+1}},
\end{equation}
where for $n\geq n'$ the integration contour crosses $\R_+$, while for $n<n'$ the contour crosses $\R_-$. This is consistent with Theorem~5.1 and Proposition~3.2 of~\cite{BF08}, where the correlation kernel was obtained by taking the finite system considered in Theorem~\ref{ThmFiniteSystem} and followed by taking the large time/space limit in such a way that the local normal direction of the surface is given by $(p_a,p_b,p_c)$\footnote{To see the exactness of the connection from the kernel in Proposition~3.2 of~\cite{BF08}, one needs to keep in mind this kernel is the limit of the one in Corollary~4.1 of~\cite{BF08} instead of the original ${\cal K}_t$: this introduced a conjugation factor and a shift in the $x$ by $n$. Finally, the point in the complex plane $\Omega$ in~\cite{BF08} equals $1-\overline\Omega_{abc}$ here.}.

\begin{thm}\label{ThmInfiniteSystem}
Consider the particle model on $\Z^2$ distributed according to $M_\nu$, that is, with determinantal correlation functions given by the correlation kernel (\ref{eqApp1}). Then the speed of growth is given by
\begin{equation}\label{eq2.17}
v=-\K^{-1}_\nu(x,n;x+1,n)=\frac{\Im(\Omega_{abc})}{\pi}= \frac{1}{\pi}\frac{\sin(\theta_b)\sin(\theta_c)}{\sin(\theta_a)}.
\end{equation}
\end{thm}

\begin{rem}
The speed of growth of the model defined in~\cite{Ton15} is then given by \mbox{$(p-q)v$}, where $p,q$ are the two parameters in the model of~\cite{Ton15}. The totally asymmetric dynamics is the one given by $p=1$ and $q=0$.
\end{rem}

\section{Proof of theorems}\label{SectProof}

\subsection{Algebraic proof of Theorem~\ref{ThmFiniteSystem}}
The algebraic proof presented here is strongly inspired by the combinatorial proof for finite graphs obtained beforehand that will be used to prove  Theorem~\ref{ThmInfiniteSystem}. In this proof we will suppress all the $t$ indices and have reasonably sized matrices, we use the notation
\begin{equation}
{\cal K}_{x,n;x',n'}:={\cal K}(x,n;x',n').
\end{equation}

We use the coordinate system illustrated in Figure~\ref{FigCoordinates}.
We derive a series expansion of ${\cal K}(x,n;x+1,n)$ by expanding it step-by-step.
The idea behind the expansion is that ${\cal K}(x,n;x+1,n)$ is (intuitively) suggestive of an ``edge'' between $\BB(x,m)$ and $\WW(x+1,m)$. If this ``edge'' is covered by a dimer, then either there is a dimer covering $(\BB(x+1,m-1),\WW(x,m))$ or not. In the latter case, then there are dimers covering $(\BB(x,m-1),\WW(x,m))$ and $(\BB(x+1,m-1),\WW(x+1,m-1))$. This can be repeated for $m=n,n-1,\ldots,1$ and for $m=1$. The latter case will not occur and therefore the series naturally ends. This idea is exploited in depth for the finite graph lozenge tiling; see Proposition~\ref{PropRecursionGraph}.

We start with two algebraic identities satisfied by the kernel (\ref{StartingKernel}).
\begin{lem}\label{LemIdentity}
It holds
\begin{equation}
{\cal K}_{x,n;x'+1,n'-1} = {\cal K}_{x,n;x'+1,n'}-{\cal K}_{x,n;x',n'} +\delta_{n,n'-1}\delta_{x'+1,x}.
\end{equation}
\end{lem}
\begin{proof}
It is quite trivial. One uses linearity of the integrals.
\end{proof}

\begin{lem}\label{LemIdentityB}
It holds
\begin{equation}
{\cal K}_{x,n-1;x',n'}-{\cal K}_{x+1,n-1;x',n'} = {\cal K}_{x,n;x',n'} +\delta_{n',n}(\delta_{x',x+1}-\delta_{x',x}).
\end{equation}
\end{lem}
\begin{proof}
It is quite trivial. One uses linearity of the integrals.
\end{proof}
We start with a proposition that will be used recursively.
\begin{prop}\label{propRecursionFiniteT}
Consider any set of $M$ $\{(\BB(x_i,m_i),\WW(x_i',m_i'),1\leq i \leq M\}$ black/white couples of (disjoint) vertices that do not include black vertices at $(x,m),(x+1,m-1),(x,m-1)$ and white at $(x+1,m),(x,m),(x+1,m-1)$. Then
\begin{equation}\label{eq3.3}\small
\begin{aligned}
& \det\left[\begin{array}{ll}
{\cal K}_{x_i,m_i;x_j',m_j'} & {\cal K}_{x_i,m_i;x+1,m} \\
{\cal K}_{x,m;x_j',m_j'} & {\cal K}_{x,m;x+1,m}
\end{array}\right]=\det\left[\begin{array}{lll}
{\cal K}_{x_i,m_i;x_j',m_j'} & {\cal K}_{x_i,m_i;x,m} & {\cal K}_{x_i,m_i;x+1,m}\\
{\cal K}_{x,m;x_j',m_j'} & {\cal K}_{x,m;x,m} & {\cal K}_{x,m;x+1,m} \\
{\cal K}_{x,m-1;x_j',m_j'} & {\cal K}_{x,m-1;x,m} & {\cal K}_{x,m-1;x+1,m}
\end{array}\right]\\
&+\det\left[\begin{array}{llll}
{\cal K}_{x_i,m_i;x_j',m_j'} & {\cal K}_{x_i,m_i;x,m} & {\cal K}_{x_i,m_i;x+1,m-1} & {\cal K}_{x_i,m_i;x+1,m}\\
{\cal K}_{x,m;x_j',m_j'} & {\cal K}_{x,m;x,m} & {\cal K}_{x,m;x+1,m-1} & {\cal K}_{x,m;x+1,m}\\
{\cal K}_{x,m-1;x_j',m_j'} & {\cal K}_{x,m-1;x,m} & {\cal K}_{x,m-1;x+1,m-1} & {\cal K}_{x,m-1;x+1,m}\\
{\cal K}_{x+1,m-1;x_j',m_j'} & {\cal K}_{x+1,m-1;x,m} & {\cal K}_{x+1,m-1;x+1,m-1} & {\cal K}_{x+1,m-1;x+1,m}
\end{array}\right],
\end{aligned}
\end{equation}
where, whenever there is a $i$ or $j$ index, this means a block-matrix with $i$ and/or $j$ from $1$ to $M$.
Schematically this is represented in Figure~\ref{FigDimerRec}, where (\ref{eq3.3}) equals (a)=(c1)+(c2).
\end{prop}

\begin{figure}
\begin{center}
\psfrag{=}[c]{$=$}
\psfrag{+}[c]{$+$}
\psfrag{-}[c]{$-$}
\psfrag{a}[c]{(a)}
\psfrag{b1}[c]{(b1)}
\psfrag{b2}[c]{(b2)}
\psfrag{c1}[c]{(c1)}
\psfrag{c2}[c]{(c2)}
\includegraphics[height=3cm]{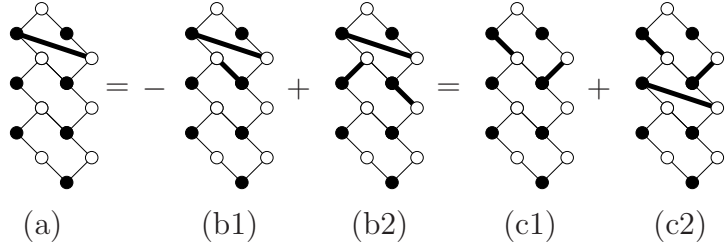}
\caption{Graphical representation of the proof of Proposition~\ref{propRecursionFiniteT}. The link in (a) is between $(\BB(x,m),\WW(x+1,m))$.}
\label{FigDimerRec}
\end{center}
\end{figure}

\begin{proof}
The left side of (\ref{eq3.3}) is represented in Figure~\ref{FigDimerRec}(a).
The scheme of Figure~\ref{FigDimerRec}(b1) is given by
\begin{equation}\label{eq3.5}
-\det\left[\begin{array}{lll}
{\cal K}_{x_i,m_i;x_j',m_j'} & {\cal K}_{x_i,m_i;x+1,m} & {\cal K}_{x_i,m_i;x,m}\\
{\cal K}_{x,m;x_j',m_j'} & {\cal K}_{x,m;x+1,m} & {\cal K}_{x,m;x,m} \\
{\cal K}_{x+1,m-1;x_j',m_j'} & {\cal K}_{x+1,m-1;x+1,m} & {\cal K}_{x+1,m-1;x,m}
\end{array}\right].
\end{equation}
The scheme of Figure~\ref{FigDimerRec}(b2) is given by
\begin{equation}
\det\left[\begin{array}{llll}
{\cal K}_{x_i,m_i;x_j',m_j'} & {\cal K}_{x_i,m_i;x+1,m} & {\cal K}_{x_i,m_i;x,m} & {\cal K}_{x_i,m_i;x+1,m-1}\\
{\cal K}_{x,m;x_j',m_j'} & {\cal K}_{x,m;x+1,m} & {\cal K}_{x,m;x,m} & {\cal K}_{x,m;x+1,m-1} \\
{\cal K}_{x,m-1;x_j',m_j'} & {\cal K}_{x,m-1;x+1,m} & {\cal K}_{x,m-1;x,m} & {\cal K}_{x,m-1;x+1,m-1} \\
{\cal K}_{x+1,m-1;x_j',m_j'} & {\cal K}_{x+1,m-1;x+1,m} & {\cal K}_{x+1,m-1;x,m} & {\cal K}_{x+1,m-1;x+1,m-1}
\end{array}\right],
\end{equation}
we subtract from the last column the previous two. Using the identity in Lemma~\ref{LemIdentity} we then obtain
\begin{equation}\label{eq3.8}
\begin{aligned}
&\det\left[\begin{array}{llll}
{\cal K}_{x_i,m_i;x_j',m_j'} & {\cal K}_{x_i,m_i;x+1,m} & {\cal K}_{x_i,m_i;x,m}& 0\\
{\cal K}_{x,m;x_j',m_j'} & {\cal K}_{x,m;x+1,m} & {\cal K}_{x,m;x,m} & 0 \\
{\cal K}_{x,m-1;x_j',m_j'} & {\cal K}_{x,m-1;x+1,m} & {\cal K}_{x,m-1;x,m} & 0 \\
{\cal K}_{x+1,m-1;x_j',m_j'} & {\cal K}_{x+1,m-1;x+1,m} & {\cal K}_{x+1,m-1;x,m} & 1
\end{array}\right]\\
&=\det\left[\begin{array}{lll}
{\cal K}_{x_i,m_i;x_j',m_j'} & {\cal K}_{x_i,m_i;x+1,m} & {\cal K}_{x_i,m_i;x,m}\\
{\cal K}_{x,m;x_j',m_j'} & {\cal K}_{x,m;x+1,m} & {\cal K}_{x,m;x,m} \\
{\cal K}_{x,m-1;x_j',m_j'} & {\cal K}_{x,m-1;x+1,m} & {\cal K}_{x,m-1;x,m}
\end{array}\right].
\end{aligned}
\end{equation}
The determinants in (\ref{eq3.5}) and (\ref{eq3.8}) differs only by the last row. Thus summing them up and using the identity in Lemma~\ref{LemIdentityB} one immediately sees that the last row and the second-last row are identical except for an extra $+1$ term in the last matrix entry. Therefore we have recovered left side of (\ref{eq3.3}). Finally one has to verify the equality between the schemes of Figure~\ref{FigDimerRec}(b1)/(b2) and Figure~\ref{FigDimerRec}(c1)/(c2). This is trivial since it corresponds to permuting the position of one column and take care of the signature of the permutation.
\end{proof}

Now we are ready to finish the proof of Theorem~\ref{ThmFiniteSystem}. The schematic representation of the proof is in Figure~\ref{FigProofFiniteT}.
\begin{figure}
\begin{center}
\psfrag{=}[c]{$=$}
\psfrag{+}[c]{$+$}
\psfrag{...}[cb]{$\ldots$}
\includegraphics[height=3cm]{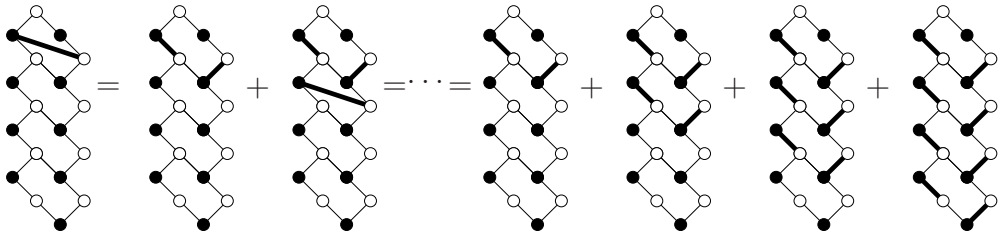}
\caption{Illustration of the recursive proof of Theorem~\ref{ThmFiniteSystem}.}
\label{FigProofFiniteT}
\end{center}
\end{figure}

Recall that $\eta(x,n)$ is the random variable of a particle located at $(x,n)$. Similarly, denote by $\sigma(x,n)$ the random variable of having a white lozenge (type III in Figure~\ref{FigureTilings})
at $(x,n)$, where the position is given by the one of the black triangle. Using Proposition~\ref{propRecursionFiniteT} repeatedly by starting with the case $m=n$ and $M=0$ and recalling the correspondence between dimers and lozenges (see Figure~\ref{FigureTilings}) we get
\begin{equation}\label{eq3.11}
{\cal K}(x,n;x+1,n)=\sum_{\ell=1}^{n} \E\left[\prod_{k=0}^{\ell-1} \eta(x,n-k)\sigma(x+1,n-k-1)\right].
\end{equation}
\begin{figure}
\begin{center}
\psfrag{forced}[l]{Forced}
\psfrag{= not particle}[l]{Only two choices}
\includegraphics[height=3cm]{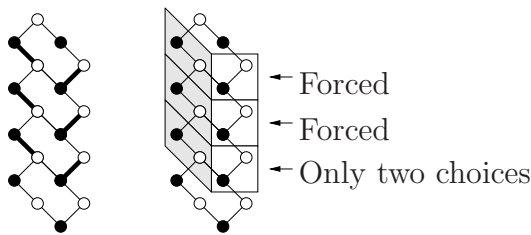}
\caption{Illustration of the connection between (\ref{eq3.11}) and (\ref{eqspeed}).}
\label{FigProofFiniteTB}
\end{center}
\end{figure}
The conditions of Theorem~\ref{ThmFiniteSystem} stipulate that the system is bounded from below by level $1$ and this is the reason why the above series is finite. Finally we have to see that (\ref{eq3.11}) and (\ref{eqspeed}) match. First consider $\ell=1,\ldots,n-1$. As illustrated in Figure~\ref{FigProofFiniteTB}, by observing that there are particles at $(x,n),\ldots,(x,n-\ell+1)$, it implies that the square lozenges at \mbox{$(x+1,n-1),\ldots,(x+1,n-\ell+1)$} occur with probability one and thus we remove them from the right side of (\ref{eq3.11}). Further, if we have a particle at $(x,n-\ell+1)$, then either there is a particle at $(x+1,n-\ell)$ or there is a white lozenge at $(x+1,n-\ell)$ since the other type of lozenge do not fit. Therefore we can replace in (\ref{eq3.11}) $\sigma(x+1,n-\ell)$ with $1-\eta(x+1,n-\ell)$. Finally, for $\ell=n$, $(x+1,0)$ is forced to be white and therefore we can remove it from the product too. This finishes the proof of Theorem~\ref{ThmFiniteSystem}.

\subsection{Proof of Theorem~\ref{ThmInfiniteSystem}}\label{Sect3.2}
The strategy is first to obtain a recursion relation analogue of Proposition~\ref{propRecursionFiniteT} for a finite honeycomb graph with generic weights, then use this result to extend the recursion relation to $\CH_L$ with $a,b,c$ weights. Using Theorem~\ref{ThmSheffield} we take the toroidal exhaustion limit and have the same recursion relation for the infinite honeycomb graph $\CH$. The proof of Theorem~\ref{ThmInfiniteSystem} ends applying recursively the recursion relation.

\subsubsection*{Recursion relation for a finite graph}
Let $G=(V,E)$ be a finite subgraph of the honeycomb graph which is tileable by dimers. $G$ is bipartite and has the same number of white and black vertices. The geometry of the honeycomb graph is as illustrated in Figure~\ref{FigCoordinates} so that the black vertices are in (a subset of) $\Z^2$ and the white vertices are in (a subset of) $(\Z+1/2)^2$. To avoid using half-integer coordinates, we adopt a notation so that the white vertices are also on $\Z^2$, more precisely, let $\mathbf{e}_1=(1/2,-1/2)$, $\mathbf{e}_2=(1/2,1/2)$, and $\mathbf{e}_3=(-1/2,1/2)$, then
\begin{equation}
\begin{aligned}
\WW(x,n)&=\BB(x,n)+\mathbf{e}_1,\\
\WW(x,n+1)&=\BB(x,n)+\mathbf{e}_2,\\
\WW(x-1,n+1)&=\BB(x,n)+\mathbf{e}_3.
\end{aligned}
\end{equation}
Edges of the graphs are of the form $\BB(x,n)$ to $\WW(x,n)$, $\WW(x,n+1)$, and \mbox{$\WW(x-1,n+1)$}. We denote by $B_G$ and $W_G$ the set of the black and white vertices respectively with the above coordinates. Assign  $\omega:E\to \R_+^*$ to be the edge weights and denote the Kasteleyn matrix, the matrix whose rows are indexed by all the white vertices and whose columns are indexed by black vertices, by
\begin{equation}
\K_G(w,b)=\left\{
\begin{array}{ll}
\omega(e),&\textrm{if }e=(w,b)\in E,\\
0,&\textrm{otherwise},
\end{array}
\right.
\end{equation}
for all $w\in W_G$ and $b\in B_G$. The above formulation defines a valid Kasteleyn orientation of the graph, that is, the number of counterclockwise edges in any face is odd; see Figure~\ref{FigKast1}.
\begin{figure}
\begin{center}
\psfrag{a}[c]{(a)}
\psfrag{b}[c]{(b)}
\includegraphics[height=2.5cm]{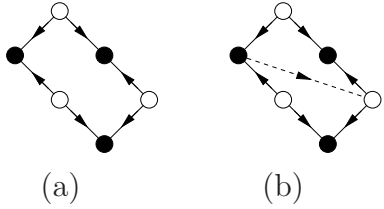}
\caption{Kasteleyn orientation of the basic honeycomb face (a) and for the face with the auxiliary edge (b).}
\label{FigKast1}
\end{center}
\end{figure}
The Kasteleyn matrix was originally introduced by Kasteleyn~\cite{Kas61} to count the number of dimer covers of a graph as the latter is given by $|\det(\K_G)|$, but its use goes beyond the uniform weight case.

The probability $\Pb$ on dimer configurations is defined as follows. For a given dimer configuration, we associate a weight to be the product of all the weights of the dimers present in that configuration. The partition function $Z_G$ is then the sum of  all weighted dimer configurations of $G$. Then, the probability of a dimer configuration is given by its weight divided by $Z_G$. In particular, given any disjoint set of edges $e_1,\ldots,e_m$, the probability of seeing dimers on the edges $e_1,\ldots,e_m$ is given by
\begin{equation} \label{eq:restrictprob}
\Pb(e_1,\ldots,e_m)=\frac{Z_{G\setminus \{e_1,\ldots,e_m\}}}{Z_G}\prod_{i=1}^m \omega(e_i).
\end{equation}
Kenyon in~\cite{Ken97} showed the following.
\begin{thm}[Kenyon~\cite{Ken97}]\label{ThmKenyon97}
Consider a set of $m$ disjoints edges of $G$, \mbox{$e_i=(w_i,b_i)\in E$}, $i=1,\ldots,m$. Then,
\begin{equation}
\Pb[e_1,\ldots,e_m]=\det\left[\K_G^{-1}(b_i,w_j)\right]_{i,j=1}^m \prod_{i=1}^m \K_G(w_i,b_i) ,
\end{equation}
where $\K_G^{-1}$ represents the inverse of $\K_G$.
\end{thm}
In other words, the dimers form a determinantal point process with correlation kernel $\mathcal{L}=\mathcal{L}(e_i,e_j)$ given by
$\mathcal{L}(e_i,e_j)=\K_G(w_i,b_i)\K_G^{-1}(b_i,w_j)$ for $e_i=(w_i,b_i)$ and $e_j=(w_j,b_j)$. The Kasteleyn matrix approach has been used with some success for computing combinatorial and asymptotics of random tiling models; see~\cite{ACJM13,CJY14,CJ14,BBCCR15} for domino tiling models and~\cite{Pet12} for the honeycomb case.

The first result is the one that inspired Proposition~\ref{propRecursionFiniteT}. To state it, we introduce some notations. For $m\geq 1$ set
\begin{equation}
\begin{aligned}
c_1(m)=&\prod_{i=0}^{m-1} \K_G(\WW(x,n-i),\BB(x,n-1-i)) \\ &\times\K_G(\WW(x+1,n-1-i),\BB(x+1,n-1-i)),\\
c_2(m)=&\prod_{i=0}^{m-1} \K_G(\WW(x,n-i),\BB(x,n-i)) \K_G(\WW(x+1,n-i),\BB(x+1,n-1-i)),\\
c_3(m)=&\frac{c_1(m)}{c_2(m+1)}\K_G(\WW(x,n-m),\BB(x+1,n-m-1))
\end{aligned}
\end{equation}
and $c_1(0)=c_2(0)=1$.

\begin{prop}\label{PropRecursionGraph}
Assume that the set of vertices
\begin{equation}
\begin{aligned}
\Sigma_m&=\{\BB(x,n-i),\WW(x+1,n-i),0\leq i\leq m\}\\
&\cup\{\BB(x+1,n-i-1),\WW(x,n-i),0\leq i\leq m-1\}
\end{aligned}
\end{equation}
belong to the graph $G$ for $0\leq m \leq N$ and let the edges
\begin{equation}
\begin{aligned}
e_i^0&=(\BB(x,n-i),\WW(x,n-i)),\\
e_i^1&=(\BB(x+1,n-i-1),\WW(x+1,n-i)),
\end{aligned}
\end{equation}
for $0\leq i \leq N$. Then
\begin{equation}\label{eq3.15}
-\K_G^{-1}(\BB(x,n),\WW(x+1,n))=\sum_{m=0}^N c_3(m)\Pb[e_0^0,e_0^1,\ldots,e_m^0,e_m^1]+R_G(N),
\end{equation}
with
\begin{equation}
R_G(N)=c_1(N+1)\frac{Z_{G\setminus\Sigma_{N+1}}}{Z_G}.
\end{equation}
Here $Z_G$ denotes the partition function of $G$ and $Z_{G\setminus\Sigma_{N+1}}$ denotes the partition function of the graph obtained from removing $\Sigma_{N+1}$ from $G$.
\end{prop}

\begin{proof}

We add an \emph{auxiliary edge} $(\BB(x,n),\WW(x+1,n))$, which is an edge not present in the graph but helpful for computations; a similar idea to that used in~\cite{CY14}. We assign a weight $1$ to the auxiliary edge $(\BB(x,n),\WW(x+1,n))$. To preserve the Kasteleyn orientation of the new graph, this edge is directed from $\BB(x,n)$ to $\WW(x+1,n)$. Since $\BB(x,n)$ and $\WW(x+1,n)$ are on the same face, removing them from the graph preserves the Kasteleyn orientation (see Figure~\ref{FigKast1}). Therefore, each matching of $G\setminus \Sigma_0$ has the same sign. Cramer's rule gives
\begin{equation} \label{eq3.16}
-\K_G^{-1}(\BB(x,n),\WW(x+1,n))=-\frac{\det[\K_{G\setminus\Sigma_0}]}{\det[\K_G]}= \frac{Z_{G\setminus \Sigma_0}}{Z_G}.
\end{equation}

For $G\setminus\Sigma_0$, $\WW(x,n)$ is either matched to $\BB(x,n-1)$ or $\BB(x+1,n-1)$. If $\WW(x,n)$ is matched to $\BB(x,n-1)$, then the edge $(\BB(x+1,n-1),\WW(x+1,n-1))$ is forced to be matched too. Notice that
\begin{equation}
\Sigma_0\cup \{\WW(x,n),\BB(x,n-1),\BB(x+1,n-1),\WW(x+1,n-1)\}=\Sigma_1.
\end{equation}
The edge weight of \mbox{$(\WW(x,n),\BB(x+1,n-1))$} is equal to \mbox{$\K_G(\WW(x,n),\BB(x+1,n-1))$} and the edge weights of $(\WW(x,n),\BB(x,n-1))$ and $(\BB(x+1,n-1),\WW(x+1,n-1))$ are equal to  $\K_G(\WW(x,n),\BB(x,n-1))$ and $\K_G(\WW(x+1,n-1),\BB(x+1,n-1))$ respectively. Notice that the product of the latter two matrix entries is exactly $c_1(1)$. Therefore,
\begin{equation}
Z_{G\setminus\Sigma_0}=
\K_G(\WW(x,n),\BB(x+1,n-1)) Z_{G\setminus(\Sigma_0\cup\{\WW(x,n),\BB(x+1,n-1)\})}
+c_1(1) Z_{G\setminus \Sigma_1}
\end{equation}
Now we proceed by induction. Consider the graph $G\setminus\Sigma_m$, which can be thought as the graph where the vertices in $\Sigma_m$ are matched as in Figure~\ref{FigKast2}(a).
\begin{figure}
\begin{center}
\psfrag{x0}[r]{$(x,n)$}
\psfrag{xm}[r]{$(x,n-m)$}
\psfrag{a}[c]{$(a)$}
\psfrag{b}[c]{$(b)$}
\psfrag{c}[c]{$(c)$}
\psfrag{d}[c]{$(d)$}
\includegraphics[height=5cm]{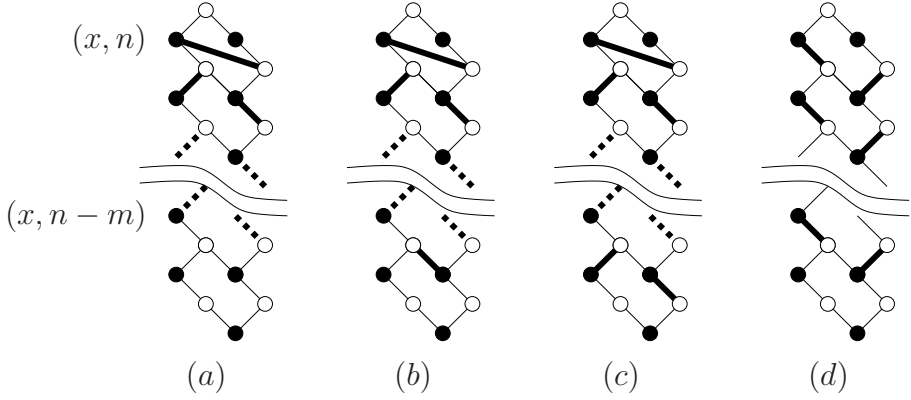}
\caption{The different graphs appearing in the proof of Proposition~\ref{PropRecursionGraph}. (a) for $\Sigma_m$, (b) for $\Sigma_m\cup\{\BB(x+1,n-m-1),\WW(x,n-m)\}$, (c) for $\Sigma_{m+1}$, and (d) for $\widetilde \Sigma_m$.}
\label{FigKast2}
\end{center}
\end{figure}
The vertex $\WW(x,n-m)$ is either matched to \mbox{$\BB(x+1,n-m-1)$} or to \mbox{$\BB(x,n-m-1)$}. For the former, then the vertex \mbox{$\BB(x+1,n-m-1)$} is incident to only one vertex, which means that also \mbox{$(\BB(x+1,n-m-1),\WW(x+1,n-m-1))$} is matched. Since
\begin{equation}
\Sigma_{m+1}=\Sigma_m\cup\{\WW(x,n-m),\BB(x,n-m-1),\BB(x+1,n-m-1),\WW(x+1,n-m-1)\},
\end{equation}
then we have
\begin{equation}\label{eq3.21}
\begin{aligned}
Z_{G\setminus \Sigma_m}=& \K_G(\WW(x,n-m),\BB(x+1,n-m-1))Z_{G\setminus (\Sigma_m\cup \{\WW(x,n-m),\BB(x+1,n-m-1)\})} \\
&+ \frac{c_1(m+1)}{c_1(m)} Z_{G\setminus \Sigma_{m+1}}.
\end{aligned}
\end{equation}
Define
\begin{equation}
\widetilde \Sigma_m=\cup_{i=0}^m\{\BB(x,n-i),\WW(x,n-i),\BB(x+1,n-i-1),\WW(x+1,n-i)\}.
\end{equation}
The vertex set $\widetilde \Sigma_m$ is equal to $\Sigma_m\cup\{\BB(x+1,n-m-1),\WW(x,n-m)\}$. Therefore (\ref{eq3.21}) reads
\begin{equation}\label{eq3.22}
Z_{G\setminus \Sigma_m}=\K_G(\WW(x,n-m),\BB(x+1,n-m-1))Z_{G\setminus \widetilde \Sigma_m} +\frac{c_1(m+1)}{c_1(m)} Z_{G\setminus \Sigma_{m+1}}.
\end{equation}
We iterate (\ref{eq3.22}) and find that
\begin{equation}\label{eq3.23}
Z_{G\setminus \Sigma_0}=\sum_{m=0}^N c_1(m)\K_G(\WW(x,n-m),\BB(x+1,n-m-1))Z_{G\setminus \widetilde \Sigma_m} +c_1(N+1) Z_{G\setminus \Sigma_{N+1}}.
\end{equation}
Now, from~\eqref{eq:restrictprob}, we have
\begin{equation}
\frac{Z_{G\setminus\widetilde\Sigma_m}}{Z_G} = \frac{\Pb[e_0^0,e_0^1,\ldots,e_m^0,e_m^1]}{c_2(m+1)},
\end{equation}
because $c_2(m+1)$ is the product of the edge weights of the edges $e_0^0,e_0^1,\ldots,e_m^0,e_m^1$. Therefore, dividing (\ref{eq3.23}) by $Z_G$, we find the right side of (\ref{eq3.15}), while the left side follows from~\eqref{eq3.16}.
\end{proof}

In order to prove Theorem~\ref{ThmInfiniteSystem} we start by considering a finite honeycomb graphs with $(a,b,c)$ weights (see Figure~\ref{FigureTilings}). As above, define $\eta(x,n)\in\{0,1\}$ to be the random variable of having a lozenge of type I at $(x,n)$. That is
\begin{equation}
\eta(x,n)=1 \Leftrightarrow \exists\textrm{ a dimer at }(\WW(x,n),\BB(x,n)).
\end{equation}
Then Proposition~\ref{PropRecursionGraph} gives the following.
\begin{cor}\label{CorFiniteabc}
For any $(a,b,c)$-weighted honeycomb graph $G$ (satisfying the assumptions of Proposition~\ref{PropRecursionGraph}),
\begin{equation}
-\frac{bc}{a}\K^{-1}_G(\BB(x,n),\WW(x+1,n))= \sum_{m=0}^N\E\bigg[(1-\eta(x+1,n-m))\prod_{i=0}^m \eta(x,n-i)\bigg] + \widetilde R_G(N)
\end{equation}
with $\widetilde R_G(N)=a^{-1}(bc)^{N+2} Z_{G\setminus\Sigma_{N+1}}/Z_G$.
\end{cor}
\begin{proof}
The $(a,b,c)$-weighting means that
\begin{equation}
\begin{aligned}
\K_G(\WW(x,n),\BB(x+1,n-1))&=a,\\
\K_G(\WW(x,n),\BB(x,n))&=b,\\
\K_G(\WW(x,n),\BB(x,n-1))&=c,
\end{aligned}
\end{equation}
see Figure~\ref{FigureTilings}. Then $c_1(m)=(bc)^m$ as well as $c_2(m)=(bc)^m$ so that Proposition~\ref{PropRecursionGraph} gives immediately
\begin{equation}
-\K^{-1}_G(\BB(x,n),\WW(x+1,n))= \sum_{m=0}^N\frac{a}{bc}\Pb[e_0^0,e_0^1,\ldots,e_m^0,e_m^1] + R_G(N)
\end{equation}
with $R_G(N)=(bc)^{N+1} Z_{G\setminus\Sigma_{N+1}}/Z_G$.

Next notice that having a chain of particles $\{\eta(x,n-i)\}_{i=0}^m$ which corresponds to the edges $e_0^0,\ldots,e_m^0$ and no particle at $(x+1,n-m)$ corresponding to the edge $(\BB(x+1,n-m-1),\WW(x+1,n-m))$, then the dimer configurations incident to the edge $\BB(x+1,n-m+i)$ for $i=0,\ldots,m-1$ are forced since each vertex becomes incident to one edge by increasing $i$. Hence the dimers \mbox{$(\BB(x+1,n-i-1),\WW(x+1,n-i))$} for $i=0,\ldots,m$ are forced. This implies
\begin{equation}\label{eq3.30}
\Pb[e_0^0,e_0^1,\ldots,e_m^0,e_m^1]=\Pb[e_0^0,\ldots,e_m^0,e_m^1]=\E\bigg[(1-\eta(x+1,n-m))\prod_{i=0}^m \eta(x,n-i)\bigg].
\end{equation}
This ends the proof of the corollary.
\end{proof}

\subsubsection*{Recursion relation on the infinite honeycomb graph $\CH$}
We focus on the $(a,b,c)$-weighting of the honeycomb graph. Recall that $\CH$ denotes the infinite honeycomb graph. The first result is the extension of the recursion relation to the infinite honeycomb graph.
\begin{prop}\label{PropRecursionInfinite}
On the infinite honeycomb graph with $(a,b,c)$-weights it holds
\begin{equation}\label{eq3.40}
-\frac{bc}{a}\K^{-1}_{abc}(\BB(x,n),\WW(x+1,n))= \sum_{m=0}^N\E\left[(1-\eta(x+1,n-m))\prod_{k=0}^m\eta(x,n-k)\right]+R_N
\end{equation}
with
\begin{equation}\label{eq3.32b}
0\leq R_N\leq C \E\left[(1-\eta(x+1,n-N))\prod_{k=0}^N\eta(x,n-k)\right]
\end{equation}
for some finite constant $C$.
\end{prop}

\begin{proof}
The first step is to provide bounds for $\tilde{R}_G(N)$ given in Corollary~\ref{CorFiniteabc}.  We have the bound $0\leq Z_{G\backslash \Sigma_{N+1}} /  Z_{G \backslash \tilde{\Sigma}_N}\leq C$. The lower bound follows because all terms are positive. The upper bound follows because the Kasteleyn orientation on the faces of the graph $G \backslash \Sigma_{N+1}$ is the same as on the faces of the graph $G\backslash \tilde{\Sigma}_N$, and so the expansions of the determinants of the corresponding Kasteleyn matrices will have matching terms (up to a constant prefactor).  We conclude that
\begin{equation}\label{boundforfinitehex}
0 \leq \tilde{R}_G(N) \leq C \Pb[e_0^0,e_0^1,\dots ,e_N^0,e_N^1].
\end{equation}

To extend the graph to the infinite plane, we set $G$ to be the box plane partition
of size $n$, whose vertices are given by
$$\{ \bullet(x_1,x_2):  -n \leq x_1 \leq n-1,-1 \leq x_2 \leq 2n-2,-1-x_1 \leq x_2 \leq 2n-2-x_1 \}$$
and
$$ \{ \circ(x_1,x_2): 1-n\leq x_1 \leq n, 0 \leq x_2 \leq 2n-1,-x_1 \leq x_2 \leq 2n-1-x_1 \};$$
see  Fig.~\ref{fig:honeycomb} for the boxed plane partition of size $3$.
\begin{figure}
\begin{center}
\includegraphics[height=6cm]{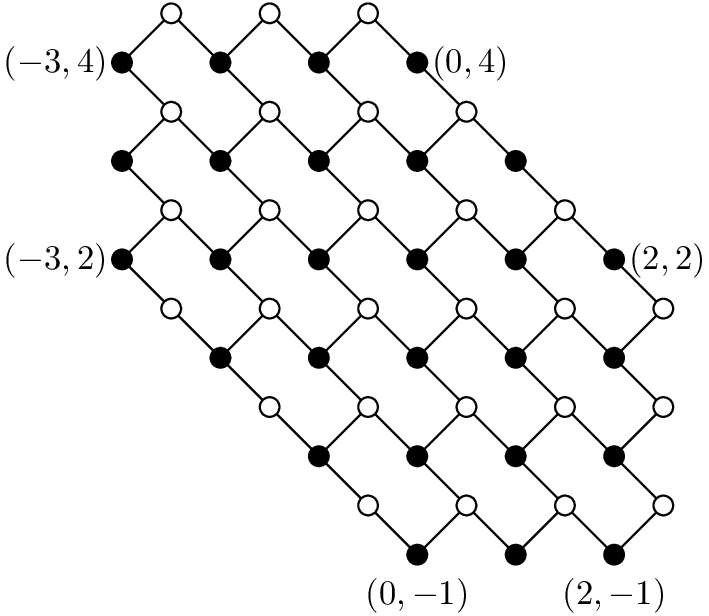}
\caption{The figure shows a boxed plane partition of size 3 with the coordinates of some black vertices. }
\label{fig:honeycomb}
\end{center}
\end{figure}

Uniformly random tilings of large box plane partitions exhibit a limit shape phenomena~\cite{CLP98}. Proposition 7.10 in \cite{Pet12} gives the bulk limit convergence of the inverse Kasteleyn matrix of the boxed plane partition of size $n$ to the infinite inverse Kasteleyn matrix with all slopes being realized. This gives the convergence of the finite-dimensional distributions as stated in Theorem 2 in \cite{Pet12}.
The above coordinates for white and black vertices agree with Petrov's~\cite{Pet12} (by setting $k=2$ and $N=2n$ in his paper). Note that Petrov's results hold for more general regions than the boxed plane partition but we do not need this level of generality here.  By applying Petrov's results to Corollary~\ref{CorFiniteabc} and~\eqref{boundforfinitehex} gives the result.

\end{proof}
\begin{rem}
The previous version of the proof of this result required embedding the graph on the torus. Unfortunately, the previous proof contained a mistake as it did not take fully into account the signs associated with the torus. Although this can be fixed, the version of the proof presented below is far simpler and avoids such problems.
\end{rem}

\subsubsection*{End of the proof of Theorem~\ref{ThmInfiniteSystem}}
With the result of Proposition~\ref{PropRecursionInfinite} we can now easily end the proof of Theorem~\ref{ThmInfiniteSystem}.

Notice that the left side of (\ref{eq3.40}) does not depend on $N$ and is finite, while the right side of (\ref{eq3.40}) is a sum of positive numbers. Therefore the last terms in the sum with $m=N$ goes to zero faster than $1/N$, which in turns implies that $R_N\to 0$ as $N\to\infty$ (in reality, the decay is exponential as shown in~\cite{Ton15}). Therefore, by taking $N\to\infty$ in (\ref{eq3.40}) we obtain
\begin{equation}
-\frac{bc}{a}\K^{-1}_{abc}(\BB(x,n),\WW(x+1,n))= \sum_{m=0}^\infty\E\left[(1-\eta(x+1,n-m))\prod_{k=0}^m\eta(x,n-k)\right].
\end{equation}
Finally notice that the prefactor in (\ref{eq2.12}) for $x'=x+1$ and $n'=n$ is exactly $bc/a$, i.e.,
\begin{equation}
\frac{bc}{a}\K^{-1}_{abc}(\BB(x,n),\WW(x+1,n)) = \K^{-1}_\nu(x,n;x+1,n).
\end{equation}
The other equalities are easy to compute and were already obtained in~\cite{BF08}.

\begin{rem}
It seems plausible that Theorem~\ref{ThmInfiniteSystem} could also be verified using Proposition~\ref{propRecursionFiniteT} recursively for the kernel $\K^{-1}_\nu(x,n;x',n')$ giving  an algebraic proof. However, a bound for the remainder, analogous to (\ref{eq3.32b}), would still be required.  This bound seems mysterious from the algebraic expression and even the positivity of the remainder is not obvious from the determinant expression.
\end{rem}

\appendix

\section{Equivalence of kernels}
Here we derive a single integral representation for the inverse Kasteleyn matrix $\K^{-1}_\nu(x,n;x',n')$ given in (\ref{eq2.14}). Let us do the change of variables $w\to -w$ so that
\begin{equation}
\K^{-1}_\nu(x,n;x',n')=\frac{(-1)^{n-n'+x-x'}}{(2\pi\I)^2}\oint_{|z|=b/a}dz\oint_{|w|=c/a}dw \frac{z^{n-n'} w^{n'-n+x'-x-1}}{1+z-w}.
\end{equation}
\emph{Case $n\geq n'$}. In this case when $w$ is not in the arc of the circle of radius $c/a$ (anticlockwise oriented) from $\overline\Omega_{abc}$ to $\Omega_{abc}$, then no poles of $z$ lies inside its integration contour and the contribution is $0$. If $w$ is in the arc of circle from $\overline\Omega_{abc}$ to $\Omega_{abc}$, then there is a simple pole at $z=w-1$, which leads to
\begin{equation}\label{eqApp1}
\K^{-1}_\nu(x,n;x',n')=\frac{(-1)^{n-n'+x-x'}}{2\pi\I}\int_{\overline\Omega_{abc}}^{\Omega_{abc}} dw \frac{(w-1)^{n-n'}}{w^{n-n'+x-x'+1}}.
\end{equation}
Here the integration path can be any path crossing the real axis on $\R_+$.

\emph{Case $n<n'$}. In this case when $w$ is in the arc of the circle of radius $c/a$ (anticlockwise oriented) from $\overline\Omega_{abc}$ to $\Omega_{abc}$, then outside the integration contour of $z$ there are no poles and the integral over $z$ gives $0$. If $w$ is in the arc of circle from $\Omega_{abc}$ to $\overline\Omega_{abc}$, then there is a simple pole at $z=w-1$, which leads to
\begin{equation}
\K^{-1}_\nu(x,n;x',n')=-\frac{(-1)^{n-n'+x-x'}}{2\pi\I}\int_{\Omega_{abc}}^{\overline\Omega_{abc}} dw \frac{(w-1)^{n-n'}}{w^{n-n'+x-x'+1}}.
\end{equation}
This is equal to the expression (\ref{eqApp1}) if the paths are chosen to cross real axis on $\R_-$.


\end{document}